\pdfoutput=1
\RequirePackage{ifpdf}
\ifpdf 
\documentclass[pdftex]{sigma}
\else
\documentclass{sigma}
\fi

\numberwithin{equation}{section}

\newtheorem{Theorem}{Theorem}[section]
\newtheorem{Lemma}[Theorem]{Lemma}
\newtheorem{Proposition}[Theorem]{Proposition}
 { \theoremstyle{definition}
\newtheorem{Example}[Theorem]{Example}
\newtheorem{Remark}[Theorem]{Remark} }

\begin{document}

\allowdisplaybreaks

\newcommand{\arXivNumber}{1706.02873}

\renewcommand{\PaperNumber}{077}

\FirstPageHeading

\ShortArticleName{Non-Homogeneous Hydrodynamic Systems and Quasi-St\"{a}ckel Hamiltonians}

\ArticleName{Non-Homogeneous Hydrodynamic Systems\\ and Quasi-St\"{a}ckel Hamiltonians}

\Author{Krzysztof MARCINIAK~$^\dag$ and Maciej B{\L}ASZAK~$^\ddag$}

\AuthorNameForHeading{K.~Marciniak and M.~B{\l}aszak}

\Address{$^\dag$~Department of Science and Technology, Campus Norrk\"{o}ping, Link\"{o}ping University, Sweden}
\EmailD{\href{krzma@itn.liu.se}{krzma@itn.liu.se}}

\Address{$^\ddag$~Faculty of Physics, Division of Mathematical Physics, A.~Mickiewicz University,\\
\hphantom{$^\ddag$}~Pozna\'{n}, Poland}
\EmailD{\href{blaszakm@amu.edu.pl}{blaszakm@amu.edu.pl}}

\ArticleDates{Received June 12, 2017, in f\/inal form September 25, 2017; Published online September 28, 2017}

\Abstract{In this paper we present a novel construction of non-homogeneous hydrodynamic equations from what we call quasi-St\"{a}ckel systems, that is non-commutatively integrable systems constructed from appropriate maximally superintegrable St\"{a}ckel systems. We describe the relations between Poisson algebras generated by quasi-St\"{a}ckel Hamiltonians and the corresponding Lie algebras of vector f\/ields of non-homogeneous hydrodynamic systems. We also apply St\"{a}ckel transform to obtain new non-homogeneous equations of considered type.}


\Classification{70H06; 70H20; 35F50; 53B20}

\section{Introduction}

In the present paper we investigate a class of f\/irst-order quasi-linear PDEs of the form
\begin{gather}
q_{t}=K(q)q_{x}+Y(q),\label{1.1}
\end{gather}
where $q=(q_{1},\dots ,q_{n})^{T},$ $q_{i}=q_{i}(x,t)$, $K$ is an $n\times n$ matrix depending on $q$ and $Y$ is a vector depending on $q$ (vector f\/ield). The systems~(\ref{1.1}) are called in literature non-homogenous hyd\-ro\-dynamic or dispersionless systems. Specif\/ically, we shall focus on the systems~(\ref{1.1}) with the homogenous part $K(q)q_{x}$ being semi-Hamiltonian in the sense of Tsarev \cite{Tsarev1991} and weakly non\-li\-near~\cite{roz}. The homogeneous case $Y=0$ has been intensively studied in literature due to the fact that such equations are strongly connected with classical integrable systems. For example, any solution of the homogeneous equation $q_{t}=K(q)q_{x}$ is also a solution of some particular Liouville integrable and separable f\/inite-dimensional system, so called St\"{a}ckel system~\cite{Zenia1991, Zenia1997}. The matrices $K(q)$ can in this context be interpreted as $(1,1)$-Killing tensors of St\"{a}ckel metrics~\cite{maciej1}.

The aim of this paper is to construct a class of non-homogeneous hydrodynamic equations of type (\ref{1.1}) that can still be related to f\/inite-dimensional integrable systems. This issue was f\/irst considered in \cite{Zenia1997a, Zenia1999} but only for the case of two integrals, where the authors searched, for a given quadratic in momenta Hamiltonian, an integral of motion with quadratic in momenta terms but also having additional terms linear in momenta (magnetic terms), generating non-homogeneous terms in the related hydrodynamic equation. They noticed that these linear in momenta terms are given by vectors that are Killing vectors for the metric of the f\/irst Hamiltonian.

Another approach to the idea of constructing Liouville integrable systems with linear in momenta terms (restricted to cases $n=2$ and $3$)~has been presented in~\cite{marikhin2005} and later in~\cite{marikhin2014}, where the authors def\/ined these systems as so called quasi-St\"{a}ckel systems and also noted some connections of such systems with non-homogenous hydrodynamic equations.

Inspired by these results we will look for a way~-- valid in any dimension~-- of modifying the St\"{a}ckel Hamiltonians by linear in momenta terms (which makes the corresponding hydrodynamic equations non-homogeneous) such that $(i)$~the modif\/ied system def\/ines a maximally superintegrable system where the extra integrals of motion are all linear in momenta and $(ii)$~it constitutes a~non-commutative Poisson algebra. As a consequence, the structure constants of this algebra and the algebra of corresponding non-homogeneous hydrodynamic vector f\/ields will be (up to a sign) the same. We call these modif\/ied systems quasi-St\"{a}ckel systems as well, but we want to stress that our def\/inition of quasi-St\"{a}ckel systems is dif\/ferent from the def\/inition in~\cite{marikhin2014}, as there the author demands these systems to be Liouville integrable while our systems are integrable only in the non-commutative sense.

The paper is organized as follows. In Section~\ref{section2} we brief\/ly remind the relation between homogeneous hydrodynamic Killing equations and geodesic St\"{a}ckel systems. In Section~\ref{section3} we restrict our attention to a particular class of St\"{a}ckel systems, namely Benenti class with constant curvature metrics (it should be pointed out that there is a substantial literature on constant curvature systems with magnetic terms, see for example~\cite{MSW} and references there). Each system of this class consists of $n$ geodesic Hamiltonians $E_{i}$. For this particular class we f\/ind $n-1$ Killing vectors $Y_{i}$ of the metric tensor of the f\/irst Hamiltonian $E_{1},$ generating additional linear in momenta integrals $W_{i}$ which yields a maximally superintegrable St\"{a}ckel system. Then we investigate the Poisson algebras generated by the Hamiltonians $h_{i}=E_{i}+W_{i}$ which we call \emph{quasi-St\"{a}ckel Hamiltonians}. Further, we show that the quasi-St\"{a}ckel Hamiltonians $h_{i}$ can be generated from a new type of linear relations that generalize the separation relations for the St\"{a}ckel Hamiltonians $E_{i}$. We call this new type of relations \emph{quasi-separation relations} (which for $n=3$ are particular realizations of the quasi-St\"{a}ckel systems from~\cite{marikhin2014}). In Section~\ref{section4} we prove a theorem (Theorem~\ref{jedyne}) which establishes an explicit relation between Poisson algebras of $h_{i}$ and Lie-algebras of the corresponding non-homogeneous hydrodynamic vector f\/ields~(\ref{1.1}). In Section~\ref{section5} we exploit the notion of St\"{a}ckel transform \cite{maciej2, maciej3, Boyer,H, Artur} and analyze which of our systems~$h_{i}$ can be mapped by this transform to new systems $\tilde{h}_{i}$ in such a way that the Hamiltonians $\tilde{h}_{i}$ also constitute an algebra. In this way we obtain new non-homogeneous hydrodynamic equations.

Let us point out that this article does not deal with the problem of integrability of the obtained non-homogeneous hydrodynamic equations. This is a~separate problem yet to investigate.

\section[Homogeneous hydrodynamic Killing equations generated by geodesic St\"{a}ckel systems]{Homogeneous hydrodynamic Killing equations\\ generated by geodesic St\"{a}ckel systems}\label{section2}

In this section we brief\/ly remind some facts about the relation between hydrodynamic Killing equations and geodesic St\"{a}ckel systems.

Consider the following $n$ separation relations
\begin{gather}
\sum_{j=1}^{n}\Phi_{ij} ( \lambda_{i} ) E_{j}=\frac{1}{2}f_{i} ( \lambda_{i} ) \mu_{i}^{2},\qquad i=1,\ldots,n,\label{1}
\end{gather}
on $M=\mathbf{R}^{2n}$, where $\Phi_{ij}$ and $f_{i}$ are arbitrary functions of one real variable. Since the $i$-th row of the matrix $\Phi(\lambda)$ depends on $\lambda_{i}$ only, it is the so called St\"{a}ckel matrix. The variables $(\lambda,\mu) =( \lambda_{1},\ldots,\lambda_{n},\mu_{1},\ldots,\mu_{n}) $ will be in the sequel referred to as position and momentum separation coordinates on the phase space~$M$. Solving the linear system (\ref{1}) with respect to~$E_{j}$ yields~$n$ quadratic in momenta functions (Hamiltonians) on~$M$
\begin{gather}
E_{r}=\frac{1}{2}\mu^{T}A_{r}\mu , \qquad r=1,\ldots,n,\label{1a}
\end{gather}
where $A_{r}$ are $n\times n$ matrices given by
\begin{gather*}
A_{r}=\operatorname{diag}\big( f_{1}(\lambda_{1}) \big( \Phi ^{-1}\big) _{r1},\ldots,f_{n} ( \lambda_{n} ) \big( \Phi ^{-1}\big) _{rn}\big), \qquad r=1,\ldots,n.
\end{gather*}
By their construction the functions $E_{r}$ are in involution with respect to the canonical Poisson bracket on $M$ given by $\{ \lambda_{i}, \mu_{j}\} =\delta_{ij}$. They are referred to in literature as \emph{geodesic St\"{a}ckel Hamiltonians}. Further, we can factorize $A_{r}$ as $A_{r}=K_{r}G$ with \begin{gather}
G=A_{1}=\operatorname{diag}\big( f_{1}(\lambda_{1}) \big( \Phi
^{-1}\big) _{11},\ldots,f_{n}(\lambda_{n}) \big( \Phi ^{-1}\big) _{1n}\big) \label{1b}
\end{gather}
and with
\begin{gather*}
K_{r}=\operatorname{diag}\left( \frac{\big(\Phi^{-1}\big) _{r1}}{\big(\Phi^{-1}\big) _{11}},\ldots,\frac{\big(\Phi^{-1}\big) _{rn}}{\big(\Phi^{-1}\big) _{1n}}\right) , \qquad r=1,\ldots ,n
\end{gather*}
(so that $K_{1}=I$). From now on we will interpret the matrix $G$ as a~contravariant form of a~metric tensor on~$M$. The corresponding covariant metric tensor will be denoted by $g$ so that $gG=I$. It can be shown that the matrices $K_{r}$ are then $(1,1)$-Killing tensors of the metric~$G$. For a~f\/ixed St\"{a}ckel matrix $\Phi$ we have thus the whole family of metrics~$G$ parametrized by~$n$ arbitrary functions $f_{i}$ of one variable. The tensors $K_{r}$ are then Killing tensors for any metric from this family. Thus, the St\"{a}ckel Hamiltonians $E_{r}$ are geodesic Hamiltonians of a Liouville integrable system in the Riemannian space $(M,g)$. The Hamiltonian equations for~$E_{r}$ are given by
\begin{gather}
\lambda_{t_{r}}=\frac{\partial E_{r}}{\partial\mu},\qquad \mu_{t_{r}}=-\frac{\partial E_{r}}{\partial\lambda},\qquad r=1,\dots ,n,\label{2}
\end{gather}
with the implicit common solution for all the positions $\lambda_{i}=\lambda_{i}(t_{1},\dots ,t_{n})$ in the form
\begin{gather}
\sum_{k=1}^{n}\int^{\lambda_{k}}\frac{\Phi_{k,n-r}(\xi)}{\varphi_{k}(\xi)}d\xi=t_{r},\qquad r=1,\dots ,n,\label{3}
\end{gather}
where
\begin{gather}
\varphi_{k}(\xi)=\left( \frac{1}{2}f_{k}(\xi){\sum_{j=1}^{n}}\Phi_{kj}(\xi)a_{j}\right) ^{\frac{1}{2}}\label{3.5}
\end{gather}
and $a_{j}$ are arbitrary constants parametrizing the Liouville tori (values of Hamiltonians $E_{r}$). The f\/irst part of Hamiltonian equations~(\ref{2})
can be explicitly written as $\lambda_{t_{r}}=K_{r}G\mu$ so that $G\mu =\lambda_{t_{1}}$ and every solution $\lambda_{i}(t_{1},\ldots,t_{n})$ of~(\ref{2}) satisf\/ies the following system of hydrodynamic equations
\begin{gather}
\lambda_{t_{r}}=K_{r}(\lambda)\lambda_{x}\equiv Z_{r}(\lambda,\lambda _{x}) , \qquad r=2,\dots ,n,\label{4}
\end{gather}
where we denoted $t_{1}=x$. We call the system~(\ref{4}) a \emph{hydrodynamic Killing system}. The dif\/ferential functions $Z_{r}$ can be interpreted as vector f\/ields on an inf\/inite-dimensional space of functions $\lambda(x)$ and since $ \{ E_{i},E_{j} \} =0$ one can show that also the vector f\/ields $Z_{r}$ commute: $ [ Z_{r},Z_{s} ] =0$ for all $r,s=1,\ldots,n$ (where $Z_{1}=\lambda_{x}$ is the translational symmetry). The system~(\ref{4}) is also known in the literature as a weakly nonlinear semi-Hamilton system and since~$K_{r}$ are diagonal in $\lambda$-variables the variables~$\lambda$ are in this context known as \emph{Riemann invariants}. The general solution of~(\ref{4}) is also given by~(\ref{3}) if we allow~$\varphi_{k}$ to be arbitrary. Therefore, due to~(\ref{3.5}), any solution of the hydrodynamic system (\ref{4}) can be obtained as a~solution of a~particular St\"{a}ckel system with the suitably chosen functions~$f_{i}$.

\section{Quasi-St\"{a}ckel systems}\label{section3}

The aim of this paper is to construct non-homogeneous hydrodynamic equations that can still be related to some f\/inite-dimensional integrable systems. The f\/irst steps in that direction was made by Ferapontov and Fordy in \cite{Zenia1997a, Zenia1999}. As we wrote in Introduction we will modify the geodesic Hamiltonians $E_{i}$ by linear in momenta constants of motion $W_{i} $ for Hamiltonian $E_{1}$, which makes the corresponding hydrodynamic equations non-homogeneous, so that $(i)$ the modif\/ied system def\/ines a maximally superintegrable system where the extra integrals of motion are all linear in momenta and $(ii)$ it constitutes a non-commutative Poisson algebra. As a~consequence, the structure constants of this algebra and the algebra of corresponding hydrodynamic vector f\/ields will be (up to a sign) the same.

It is well known that the linear in momenta function $W=p^{T}Y= \sum_{i} p_{i}Y^{i}$ on $M$ is an integral of motion of $E_{1}$ if and only if $Y=\sum_{i} Y^{i}(q)\frac{\partial}{\partial q_{i}}$ is a Killing vector for the metric~$G$. Therefore, we have to f\/ind enough Killing vectors $Y_{r}$ of $G$ in order to perform our task. In analogy with the homogeneous case, the f\/irst part of Hamiltonian equations~(\ref{2}) for the modif\/ied Hamiltonians $h_{r}=E_{r}+W_{r}$ will then take the non-homogenous form
\begin{gather}
\lambda_{t_{r}}=K_{r}(\lambda)\lambda_{x}+Y_{r}(\lambda)\equiv Z_{r} (\lambda,\lambda_{x}) , \qquad r=2,\dots ,n,\label{4a}
\end{gather}
where $W_{r}=p^{T}Y_{r}$.

We will solely work with the metrics of constant curvature. It is well known that the Lie algebra of Killing vectors of constant curvature metrics is of the maximal dimension $n(n+1)/2$, but the problem of identifying all the constant-curvature metrics of the form (\ref{1b}) is not solved yet, so we will focus on some particular classes of metrics of the form~(\ref{1b}) which are known to be of constant curvature.

\subsection{Benenti class of St\"{a}ckel systems}

We will thus impose two restrictions on the metric $G$. Firstly, we will only consider a class of St\"{a}ckel systems (\ref{1}) given by the St\"{a}ckel matrix of the very particular form $\Phi_{ij}=\lambda_{i}^{n-j}$. This results in the following separation relations
\begin{gather}
\sum\limits_{j=1}^{n} \lambda_{i}^{n-j}E_{j}=\frac{1}{2}f_{i}(\lambda_{i})\mu_{i}^{2},\qquad i=1,\ldots,n.\label{Ben}
\end{gather}
Such systems are called in literature \emph{Benenti systems}. Moreover, in this case the metric tensor~(\ref{1b}) attains the explicit form
\begin{gather*}
G=A_{1}=\operatorname{diag}\left( \frac{f_{1}(\lambda_{1}) }{\Delta _{1}},\ldots,\frac{f_{n}(\lambda_{n}) }{\Delta_{n}}\right)
,\qquad \Delta_{i}=\prod_{j\neq i}^{n}\left( \lambda_{i}-\lambda_{j}\right), 
\end{gather*}
while the Killing tensors $K_{r}$ are
\begin{gather}
K_{r}=(-1)^{r+1}\operatorname*{diag}\left( \frac{\partial\sigma_{r}}{\partial\lambda_{1}},\dots,\frac{\partial\sigma_{r}}{\partial\lambda_{n}}\right),
 \qquad r=1,\ldots,n,\label{Kr}
\end{gather}
where $\sigma_{r}(\lambda)$ are elementary symmetric polynomials in $\lambda$.

Secondly, we will also assume that all $f_{i}$ are equal to the same monomial of order not exceeding $n+1$
\begin{gather*}
f_{i}(\lambda_{i})=\lambda_{i}^{m}, \qquad m\in \{ 0,\ldots,n+1\},
\end{gather*}
which renders $n+2$ metrics
\begin{gather*}
G_{m}=\operatorname{diag}\left( \frac{\lambda_{1}^{m}}{\Delta_{1}},\ldots,\frac{\lambda_{n}^{m}}{\Delta_{n}}\right),\qquad m\in\left\{
0,\ldots,n+1\right\},
\end{gather*}
with the corresponding geodesic St\"{a}ckel Hamiltonians
\begin{gather}
E_{r}^{m}=\frac{1}{2}\mu^{T}K_{r}G_{m}\mu, \qquad r=1,\ldots,n.\label{erm}
\end{gather}

\begin{Remark}It can be shown that the metric $G_{m}$ is f\/lat for $m\in \{0,\ldots,n\} $ and of constant curvature for $m=n+1$.
The separation variables $\lambda_{i}$ can in this case be considered as appropriate degenerations of Jacobi elliptic coordinates, see~\cite{KM}.
\end{Remark}

In the separation coordinates (Riemann invariants) $\lambda$ the hydrodynamic equations~(\ref{4}) take a particular form
\begin{gather}
\frac{d\lambda}{dt_{r}}=(-1)^{r+1}\frac{\partial\sigma_{r}}{\partial\lambda }\lambda_{x},\qquad r=1,\ldots,n.\label{Kill}
\end{gather}
It turns out that the search for the linear in momenta integrals for $E_{1}$ is much easier in the coordinates
\begin{gather}
q_{i}(\lambda)=(-1)^{i}\sigma_{i}(\lambda).\label{q}
\end{gather}
We will refer to them as Vi\`{e}te coordinates. The corresponding conjugate momenta $p(\lambda,\mu)$ on~$M$ are given by
\begin{gather}
p_{i}=-\sum_{k=1}^{n}\frac{\lambda_{k}^{n-i}}{\Delta_{k}}\mu_{k}, \qquad i=1,\ldots,n.\label{p}
\end{gather}
In $q$-variables the components of $G_{m}$ are polynomials in $q$ given by~\cite{maciej4}
\begin{gather}
\left( G_{0}\right) ^{ij} =\sum_{k=0}^{n-1}q^{k}\delta_{n+k+1}^{i+j},\nonumber\\
\left( G_{m}\right) ^{ij} =
\begin{cases}
\displaystyle \sum\limits_{k=0}^{n-m-1}q_{k}\delta_{n-m+k+1}^{i+j},& i,j=1,\dots ,n-m,\\
\displaystyle-\sum\limits_{k=n-m+1}^{n}q_{k}\delta_{n-m+k+1}^{i+j}, & i,j=n-m+1,\dots ,n,\\
0, & \text{otherwise},
\end{cases} \qquad m=1,\dots ,n,\label{Gr}\\
( G_{n+1}) ^{ij} =q_{i}q_{j}-q_{i+j},\qquad i,j=1,\dots ,n\nonumber
\end{gather}
(where we denote $q_{0}=1$), while the hydrodynamic equations (\ref{4}) attain the form $q_{t_{r}}=K_{r}(q)q_{x}$ or, explicitly~\cite{maciej5}
\begin{gather}
\frac{d q_{j}}{dt_{r}}=( Z_{r}) ^{j}=(q_{j+r-1})_{x}+\sum
_{k=1}^{j-1}\big( q_{k}( q_{j+r-k-1}) _{x}-q_{j+r-k-1}(q_{k}) _{x}\big), \qquad r,j=1,\ldots,n.\!\!\!\!\label{Kilq}
\end{gather}

\subsection{Killing vectors and additional integrals of motion}

Let us def\/ine the following sets of indices:
\begin{gather*}
I_{1}^{m}= \{ 2,\ldots,n-m+1 \}, \qquad I_{2}^{m}= \{n-m+2,\ldots,n \}, \qquad m=0,\ldots,n+1.
\end{gather*}
The search for linear in momenta extra integrals of motion for~$E_{1}^{m}$ that are also linear in $q$ yields the following lemma.

\begin{Lemma}\label{W}For $m=0,\ldots,n-1$ the functions
\begin{gather}
W_{r}^{m}=\sum_{i=1}^{r-1}iq_{r-i-1}p_{n-m-i+1}, \qquad r\in I_{1}^{m},\label{W1}
\end{gather}
and for $m=2,\ldots,n+1$ the functions
\begin{gather}
W_{r}^{m}=\sum_{i=1}^{n-r+1}iq_{r+i-1}p_{n-m+i+1}, \qquad r\in I_{2}^{m}\label{W2}
\end{gather}
are constants of motion for $E_{1}^{m}$.
\end{Lemma}

For each $m$ between $1$ and $n$, this lemma yields $n-1$ (and $n$ in case $m=0$ and $m=n+1$) linear in both momenta $p$ and positions $q$ additional constants of motion for $E_{1}^{m}$ with the corresponding Killing vector f\/ields $Y_{r}^{m}$ for $G_{m}$ so that
\begin{gather*}
W_{r}^{m}=p^{T}Y_{r}^{m}=\sum_{i=1}^{n}p_{i}\left( Y_{r}^{m}\right) ^{i}.
\end{gather*}

\begin{Lemma}For each $m\in\{ 0,\ldots,n+1\} $ the $2n-1$ functions
\begin{gather}
E_{1}^{m},\ldots, E_{n}^{m},W_{2}^{m},\ldots,W_{n}^{m}\label{system}
\end{gather}
constitute a maximally superintegrable system with respect to $E_{1}^{m}$.
\end{Lemma}

\begin{proof}
It is well known that if a St\"{a}ckel metric $G(q)$ on $Q\subset R^{n}$ is of constant curvature then there exists $n-1$ independent second-order contravariant Killing tensors $A_{2}(q),\dots ,A_{n}(q)$ of metric $G$, which commute with respect to the Schouten bracket~(\ref{Schouten}): $[A_{i},A_{j}]_{S}=0,\ i,j=2,\dots ,n$ and there exists $\frac{1}{2}n(n+1)$ Killing vectors $Y_{i}(q)$ of~$G$. In consequence, on the phase space $T^{\ast}Q$,
there exist $n$ functions quadratic in momenta $E_{1}=\frac{1}{2}p^{T}G(q)p,E_{2}=\frac{1}{2}p^{T}A_{2}(q)p,\dots ,E_{n}=\frac{1}{2}p^{T}A_{n}(q)p$ in involution with respect to the canonical Poisson bi-vector and $\frac{1}{2}n(n+1)$ linear in momenta functions $W_{i}=p^{T}Y_{i}(q)$, forming a Poisson algebra. In the set $(E,W)$ of $n+n(n+1)/2$ functions there exist two types of functional relations: 1)~all $E_{i}$ are expressible by $W$-functions, i.e., $E_{i}=E_{i}(W)$, $i=1,\dots ,n$ and 2~$\varphi_{k}(W)=0$, where $\varphi_{k}$ is an appropriate number (depending on~$n$) of functions, quadratic in all~$W$. Now, in order to prove the functional independence of the subset~(\ref{system}), it is suf\/f\/icient to show than none of these relations survive in the set~(\ref{system}). The functional independence of $W_{2},\dots ,W_{n}$ follows from the fact that the matrix of dif\/ferentials $dW_{2}^{m},\ldots,dW_{n}^{m}$ has for each $m$ the maximal rank $n-1$. Further, since for each~$m$ all $W_{i}^{m}$, $i=2,\dots ,n$ contain together less number of $q^{\prime}s$ then any $E_{j}$ does, so $E_{j}\neq E_{j}(W_{2},\dots ,W_{n})$, $j=1,\dots ,n$.
\end{proof}

From (\ref{W1}), (\ref{W2}) and (\ref{q}), (\ref{p}) it follows that in the Riemann invariants $\lambda$ the corresponding Killing vectors $Y_{r}^{m}$ of metrics $G_{m}$ are represented by the formulas
\begin{gather}
\big( Y_{r}^{m}\big) ^{i}=\sum\limits_{k=1}^{r-1}
(-1)^{r-k} k\sigma_{r-k-1}\frac{\lambda_{i}^{m+k-1}}{\Delta_{i}},\qquad r\in I_{1}^{m}\label{Y1}
\end{gather}
and
\begin{gather}
\big( Y_{r}^{m}\big) ^{i}=\sum\limits_{k=1}^{n-r+1}
(-1)^{r+k}k \sigma_{r+k-1}\frac{\lambda_{i}^{m-k-1}}{\Delta_{i}},\qquad r\in I_{2}^{m}.\label{Y2}
\end{gather}

An important observation is that both $E_{r}^{m}$ and $W_{r}^{m}$ are homogeneous functions of the same order $n-m-r+2$ with respect to the scaling
\begin{gather*}
p_{k}\rightarrow\varepsilon^{k}p_{k} ,\qquad q_{k}\rightarrow \varepsilon^{-k}q_{k},
\end{gather*}
so it is natural to consider their sums when constructing Hamiltonians with linear in momenta terms. Let us thus def\/ine, for each~$m$, the~$n$ functions
\begin{gather}
h_{1}^{m}=E_{1}^{m},\qquad h_{r}^{m}=E_{r}^{m}+W_{r}^{m}, \qquad r=2,\ldots,n.\label{h}
\end{gather}

As we mentioned in introduction, we call the functions $h_{r}^{m}$ \emph{quasi-St\"{a}ckel Hamiltonians}.

\begin{Theorem}The functions $h_{r}^{m}$, for each fixed $m$, constitute a Poisson algebra $\mathfrak{g}=\operatorname{Span}\{ h_{r}^{m}\colon$ $r=1,\ldots,n \} $ with the following commutation relations for $i,j=2,\ldots,n$:
\begin{gather}
\big\{ h_{i}^{m},h_{j}^{m}\big\} =
\begin{cases}
0, & \text{for }i\in I_{1}^{m}\text{ and } j\in I_{2}^{m},\\
(j-i)h_{i+j-(n-m+2)}^{m}, & \text{for }i,j\in I_{1}^{m},\\
-(j-i)h_{i+j-(n-m+2)}^{m}, & \text{for }i,j\in I_{2}^{m},
\end{cases} \label{str}
\end{gather}
where we use the notation $h_{i}^{m}=0$ for $i\leq0$ or $i>n$.
\end{Theorem}

Note that for the cases $n=2$ and $n=3$ this algebra is commutative.

\begin{proof} We have
\begin{gather}
\big\{ h_{i}^{m},h_{j}^{m}\big\} =\big\{ E_{i}^{m},E_{j}^{m}\big\} +\big\{ E_{i}^{m},W_{j}^{m}\big\} +\big\{ W_{i}^{m},E_{j}^{m}\big\} +\big\{ W_{i}^{m},W_{j}^{m}\big\}, \label{ziuta}
\end{gather}
where the f\/irst term is zero as the St\"{a}ckel Hamiltonians $E_{i}^{m}$ pairwise commute. By a direct calculation one can show that the functions~$W_{i}^{m}$ themselves constitute a Poisson algebra with the same structure constants as in~(\ref{str}). Finally, a calculation involving the formulas~(\ref{Gr}) and~(\ref{Kilq}) shows that for $i,j=2,\ldots,n$
\begin{gather}
\big\{ E_{i}^{m},W_{j}^{m}\big\} +\big\{ W_{i}^{m},E_{j}^{m}\big\}
=
\begin{cases}
0, & \text{for }i\in I_{1}^{m}\text{ and }j\in I_{2}^{m},\\
(j-i)E_{i+j-(n-m+2)}^{m}, & \text{for }i,j\in I_{1}^{m},\\
-(j-i)E_{i+j-(n-m+2)}^{m}, & \text{for }i,j\in I_{2}^{m}.
\end{cases} \label{szcz}
\end{gather}
In consequence, (\ref{ziuta}) yields the right-hand side of (\ref{str}).
\end{proof}

Due to the form of~(\ref{str}) the algebra $\mathfrak{g}$ splits into a direct sum of two algebras $\mathfrak{g}=\mathfrak{g}_{1}\oplus \mathfrak{g}_{2}$ where $\mathfrak{g}_{1}=\operatorname{Span}\{ h_{i}\colon i\in I_{1}^{m} \} $ and $\mathfrak{g}_{2}=\operatorname{Span}\{ h_{i}\colon i\in I_{2}^{m}\}$.

\begin{Remark}The functions $h_{1}^{m},\ldots,h_{n}^{m},W_{2}^{m},\ldots,W_{n}^{m}$ constitute a non-commutative integrable system \cite{bolsinov2003, mf1978}.
\end{Remark}

\subsection{Quasi-separation relations}

Let us consider the following linear relations (cf.~(\ref{1}))
\begin{gather}
{\displaystyle\sum\limits_{j=1}^{n}}
\Phi_{ij}\left( \lambda_{i}\right) h_{j}=\frac{1}{2}f_{i}(\lambda_{i})\mu_{i}^{2}+\sum_{k=1}^{n}u_{ik}(\lambda)\mu_{k} , \qquad
i=1,\ldots,n,\label{qS}
\end{gather}
with a given St\"{a}ckel matrix $\Phi$, with $n$ arbitrary functions $f_{i}$ of one variable and with a~set of~$n^{2}$ functions $u_{ik}(\lambda)$ depending in general on all~$\lambda_{j}$. Similar relations were previously studied in~\cite{marikhin2014, marikhin2005}. Solving~(\ref{qS}) with respect to $h_{i}$ we obtain
\begin{gather*}
h_{r}=E_{r}+W_{r}, \qquad r=1,\ldots,n,
\end{gather*}
such that $E_{r}$ are the geodesic St\"{a}ckel Hamiltonians~(\ref{1a}), so they are in mutual involution, and $W_{r}=p^{T}Y_{r}$ are some linear in momenta terms (magnetic terms).

Assuming, additionally, that $(i)$ $W_{1}=0$ and $(ii)$ $\{ h_{1},h_{j}\} =0$ for $j=2,\ldots,n $ we obtain that $Y_{r}$ are Killing vectors of the metric~$G$ in $E_{1}$. Thus, we obtain that $E_{1},\ldots,E_{n},W_{2},\ldots,W_{n}$ constitute a maximally superintegrable and separable system. In such a case the relations~(\ref{qS}) will be called \emph{quasi-separation relations}.

\begin{Theorem}
The maximally superintegrable system \eqref{system} is generated by the following quasi-separation relations:
\begin{gather}
\sum_{j=1}^{n}\lambda_{i}^{n-j}h_{j}=\frac{1}{2}\lambda_{i}^{m}\mu_{i}^{2}+\sum_{k=1}^{n}u_{ik}(\lambda)\mu_{k},\qquad i=1,\ldots,n,\label{ola}
\end{gather}
where
\begin{gather}
\sum_{k=1}^{n}u_{ik}(\lambda)\mu_{k}=
\begin{cases}
\displaystyle -\sum_{k\neq i}\frac{\mu_{i}-\mu_{k}}{\lambda_{i}-\lambda_{k}}, & \text{for } m=0,\\
\displaystyle
-\lambda_{i}^{m-1}\sum_{k\neq i}\frac{\lambda_{i}\mu_{i}-\lambda_{k}\mu_{k}}{\lambda_{i}-\lambda_{k}}+(m-1)\lambda_{i}^{m-1}\mu_{i},
& \text{for } m=1,\ldots,n,\\
\displaystyle -\lambda_{i}^{n-1}\sum_{k\neq i}\frac{\lambda_{i}^{2}\mu_{i}-\lambda_{k}^{2}\mu_{k}}{\lambda_{i}-\lambda_{k}}+(n-1)\lambda_{i}^{n}\mu_{i}, &
\text{for }m=n+1.
\end{cases} \label{basia}
\end{gather}
\end{Theorem}

Thus, by solving (\ref{ola}), (\ref{basia}) with respect to $h_{j}$ we obtain all the Hamiltonians~(\ref{erm}) as well as the additional constants
$W_{i}^{m}$ (\ref{W1}), (\ref{W2}) of the Hamiltonian $E_{1}^{m}$.

\section{Non-homogeneous hydrodynamic equations of Killing type}\label{section4}

In this section we prove a theorem describing the relation between the systems presented in Section~\ref{section3} and non-homogeneous hydrodynamic equations. Consider the set of Hamiltonians on~$\mathbf{R}^{2n}$
\begin{gather}
h_{r}=E_{r}+W_{r}, \qquad E_{r}=\tfrac{1}{2}p^{T}A_{r}p,\qquad W_{r}=p^{T}Y_{r},\qquad r=1,\dots ,n,\label{4.1}
\end{gather}
with
\begin{gather*}
\{E_{r},E_{s}\}=0 , \qquad r,s=1,\ldots,n, 
\end{gather*}
where $A_{r}=K_{r}G$ with $L_{Y_{r}}G=0$ (i.e., all $Y_{r}$ are Killing vectors for the metric~$G$). Suppose also that the functions~$W_{r}$ constitute a~non-commutative Poisson algebra with some structure constants~$c_{rs}^{i}$
\begin{gather}
\{W_{r},W_{s}\}=\sum_{i=1}^{n}c_{rs}^{i}W_{i},\label{a}
\end{gather}
and moreover that the following condition holds
\begin{gather}
\{E_{r},W_{s}\}+\{W_{r},E_{s}\}=\sum_{i=1}^{n}c_{rs}^{i}E_{i}.\label{A}
\end{gather}
The conditions (\ref{a}) and (\ref{A}) imply that $h_{r}$ also constitute a~non-commutative Poisson algebra with the same structure constants~$c_{rs}^{i}$
\begin{gather}
\{h_{r},h_{s}\}=\sum_{i=1}^{n}c_{rs}^{i}h_{i}.\label{alg}
\end{gather}

In order to relate the algebra (\ref{alg}) with an appropriate algebra of hydrodynamic vector f\/ields we will use the link between the canonical Poisson bracket and the Schouten bracket between symmetric contravariant tensors. Actually, for a pair of functions on~$M$
\begin{gather*}
F_{K}=\frac{1}{k!}K^{i_{1}\dots i_{k}}(q)p_{i_{1}}\cdots p_{i_{k}},\qquad F_{R}=\frac{1}{r!}R^{i_{1}\dots i_{r}}(q)p_{i_{1}}\cdots p_{i_{r}}
\end{gather*}
(we use for the moment the Einstein summation convention) the following relation holds~\cite{Dolan}
\begin{gather}
\{F_{K},F_{R}\}=-([K,R]_{S})^{i_{1}\dots i_{k+r-1}}p_{i_{1}}\cdots p_{i_{k+r-1}},\label{relacja}
\end{gather}
where
\begin{gather}
([K,R]_{S})^{l_{1}\dots l_{k+r-1}}=\frac{1}{k!r!}\big[ kK^{i(l_{1}\dots }\partial_{n}R^{\dots l_{k+r-1})}-rR^{i(l_{1}\dots }\partial_{i}K^{\dots l_{k+r-1})}\big] ,\qquad \partial_{i}=\frac{\partial}{\partial q_{i}}\label{Schouten}
\end{gather}
def\/ines the Schouten bracket on $\mathbf{R}^{n}$ (the conf\/iguration space) and where $(\dots )$ is the symmetrization operation over the indices.

Formula (\ref{relacja}) implies that the vector f\/ields $Y_{r}$ constitute a~non-abelian Lie algebra with the structure constants $-c_{rs}^{i}$
\begin{gather}
[ Y_{r},Y_{s}] =-\sum_{i=1}^{n}c_{rs}^{i}Y_{i}\label{4.1c}
\end{gather}
and that
\begin{gather}
[ A_{r},Y_{s}] _{S}+[Y_{r},A_{s}]_{S}=-\sum_{i=1}^{n}c_{rs}^{i}A_{i}.\label{4.3}
\end{gather}
From the properties of the Schouten bracket it follows that $[Y_{r},A_{s}]_{S}=L_{Y_{r}}A_{s}$. Since moreover $L_{Y_{s}}A_{r}=( L_{Y_{s}}K_{r}) G$, as $Y_{s}$ are Killing vectors of a non-singular metric~$G$, we obtain from~(\ref{4.3})
\begin{gather}
L_{Y_{r}}K_{s}-L_{Y_{s}}K_{r}=-\sum_{i=1}^{n}c_{rs}^{i}K_{i}.\label{4.4}
\end{gather}

\begin{Theorem}\label{jedyne}Consider the non-homogeneous hydrodynamic systems
\begin{gather}
q_{t_{r}}=K_{r}q_{x}+Y_{r}\equiv Z_{r}\qquad r=2,\ldots,n,\label{H}
\end{gather}
where $K_{r}$ and $Y_{r}$ are such that the conditions \eqref{4.1}--\eqref{A} are satisfied. Then
\begin{gather*}
[ Z_{r},Z_{s}] =-\sum_{i=1}^{n}c_{rs}^{i}Z_{i}.
\end{gather*}
\end{Theorem}

\begin{proof}We have
\begin{gather*}
[Z_{r},Z_{s} ] =[ K_{r}q_{x}+Y_{r},K_{s}q_{x}+Y_{s}]=[ K_{r}q_{x},K_{s}q_{x}] +[ K_{r}q_{x},Y_{s}]+[ Y_{r},K_{s}q_{x}] +[ Y_{r},Y_{s}].
\end{gather*}
To begin with, $\left[ K_{r}q_{x},K_{s}q_{x}\right] =0$ for all $r$, $s$ (see for example~\cite{Zenia1997}). Further, as $[Y_{s},q_{x}]=0$ we have $[Y_{s},K_{r}q_{x}] = ( L_{Y_{s}}K_{r} ) q_{x}$. Using~(\ref{4.1c}) and~(\ref{4.4}) we obtain
\begin{gather*}
[ Z_{r},Z_{s}] =( L_{Y_{r}}K_{s}-L_{Y_{s}}K_{r})q_{x}-\sum_{i=1}^{n}c_{rs}^{i}Y_{i}=-\sum_{i=1}^{n}c_{rs}^{i}(K_{i}q_{x}+Y_{i}) =-\sum_{i=1}^{n}c_{rs}^{i}Z_{i}.\tag*{\qed}
\end{gather*}\renewcommand{\qed}{}
\end{proof}

Thus, the vector f\/ields $Z_{r}$ in (\ref{H}) constitute a Lie algebra with up to a sign the same structure constants as the Poisson algebra~(\ref{alg}). Besides, since our systems~(\ref{h}) with structure constants~(\ref{str}) satisfy the conditions (\ref{4.1})--(\ref{A}) (as~(\ref{szcz}) is a~specif\/ication of~(\ref{A})) we see that they possess non-homogeneous hydrodynamic counterparts, with the same structure constants.

\begin{Example}Let us consider the case $n=4$ in Vi\`{e}te coordinates $q$. As $m=0,\ldots,n+1$ we have then $n+2=6$ dif\/ferent non-homogeneous hydrodynamic systems~(\ref{H}):
\begin{gather*}
\left[
\begin{matrix}
q_{1}\\
q_{2}\\
q_{3}\\
q_{4}
\end{matrix}
\right] _{t_{2}}=\left[
\begin{matrix}
0 & 1 & 0 & 0\\
-q_{2} & q_{1} & 1 & 0\\
-q_{3} & 0 & q_{1} & 1\\
-q_{4} & 0 & 0 & q_{1}
\end{matrix}
\right] \left[
\begin{matrix}
q_{1}\\
q_{2}\\
q_{3}\\
q_{4}
\end{matrix}
\right] _{x}+Y_{2}\equiv Z_{2},
\\
\left[
\begin{matrix}
q_{1}\\
q_{2}\\
q_{3}\\
q_{4}
\end{matrix}
\right] _{t_{3}}=\left[
\begin{matrix}
0 & 0 & 1 & 0\\
-q_{3} & 0 & q_{1} & 1\\
-q_{4} & -q_{3} & q_{2} & q_{1}\\
-0 & -q_{4} & 0 & q_{2}
\end{matrix}
\right] \left[
\begin{matrix}
q_{1}\\
q_{2}\\
q_{3}\\
q_{4}
\end{matrix}
\right] _{x}+Y_{3}\equiv Z_{3},
\\
\left[
\begin{matrix}
q_{1}\\
q_{2}\\
q_{3}\\
q_{4}
\end{matrix}
\right] _{t_{4}}=\left[
\begin{matrix}
0 & 0 & 0 & 1\\
-q_{4} & 0 & 0 & q_{1}\\
0 & -q_{4} & 0 & q_{2}\\
0 & 0 & -q_{4} & q_{3}
\end{matrix}
\right] \left[
\begin{matrix}
q_{1}\\
q_{2}\\
q_{3}\\
q_{4}
\end{matrix}
\right] _{x}+Y_{4}\equiv Z_{4},
\end{gather*}
where
\begin{alignat*}{3}
& \text{for }m =0\colon \quad && Y_{2}=\left( 0,0,0,1\right) ^{T},\qquad Y_{3}
=(0,0,2,q_{1})^{T},\qquad Y_{4}=(0,3,2q_{1},q_{2})^{T}, &\\
& \text{for }m =1\colon \quad && Y_{2}=(0,0,1,0)^{T},\qquad Y_{3}=(0,2,q_{1},0)^{T},\qquad Y_{4}=(3,2q_{1},q_{2},0)^{T}, & \\
& \text{for }m =2\colon \quad && Y_{2}=(0,1,0,0)^{T},\qquad Y_{3}=(2,q_{1},0,0)^{T},\qquad Y_{4}=(0,0,0,q_{4})^{T}, &\\
& \text{for }m =3\colon \quad && Y_{2}=(1,0,0,0)^{T},\qquad Y_{3}=(0,0,q_{3},2q_{4}),\qquad Y_{4}=(0,0,q_{4},0)^{T}, & \\
& \text{for }m =4\colon \quad && Y_{2}=(0,q_{2},2q_{3},3q_{4})^{T},\qquad Y_{3}=(0,q_{3},2q_{4},0)^{T},\qquad Y_{4}=(0,q_{4},0,0)^{T}, &\\
& \text{for }m =5\colon \quad && Y_{2}=(q_{2},2q_{3},3q_{4},0)^{T},\qquad Y_{3}=(q_{3},2q_{4},0,0)^{T},\qquad Y_{4}=(q_{4},0,0,0)^{T},&
\end{alignat*}
which constitute appropriate Lie algebras with the following nonzero elements given by Theorem~\ref{jedyne} and~(\ref{str})
\begin{alignat*}{3}
& \text{for }m=0\colon \quad && [Z_{3},Z_{4}]=-Z_{1},& \\
& \text{for }m=1\colon \quad && [Z_{2},Z_{4}]=-2Z_{1},\qquad [Z_{3},Z_{4}]=-Z_{2},&\\
&\text{for }m=2\colon \quad && [Z_{2},Z_{3}]=-Z_{1},&\\
&\text{for }m=3\colon \quad && [Z_{3},Z_{4}]=Z_{4},&\\
&\text{for }m=4\colon \quad && [Z_{2},Z_{3}]=Z_{3},\qquad [Z_{2},Z_{4}]=2Z_{4},&\\
&\text{for }m=5\colon \quad && [Z_{2},Z_{3}]=Z_{3},&
\end{alignat*}
where $Z_{1}=(q_{1},q_{2},q_{3},q_{4})_{x}^{T}$. Using (\ref{Kill}), (\ref{Y1}) and (\ref{Y2}) the present example can be easily calculated in the Riemman invariants~$\lambda$.
\end{Example}

\section{St\"{a}ckel transform and new non-homogeneous hydrodynamic Killing systems}\label{section5}

St\"{a}ckel transform is a functional transform that maps a Liouville integrable system into a new integrable system, and in particular it maps a~St\"{a}ckel system into a new St\"{a}ckel systems \cite{maciej2,Boyer, H, Artur}, which explains its name. In~\cite{maciej3} the authors considered the action of St\"{a}ckel transform on superintegrable systems in such a way that it preserves superintegrability. It was found that only particular one-parameter St\"{a}ckel transforms preserve superintegrability. Here we demonstrate that St\"{a}ckel transforms are also applicable for our particular systems~(\ref{h}) def\/ined by the quasi-separation relations (\ref{ola}). Nevertheless, if we demand that the transformed system shall also constitute a~Poisson algebra, then the number of admissible St\"{a}ckel transforms becomes very limited.

Consider thus Hamiltonians (\ref{h}) extended by some potentials $V_{r}^{(k)}$
\begin{gather}
h_{1}^{m}=E_{1}^{m}+\alpha V_{1}^{(k)}\equiv H_{1}^{m}, \nonumber\\
h_{r}^{m}=E_{r}^{m}+\alpha V_{r}^{(k)}+W_{r}^{m}\equiv H_{r}^{m}+W_{r}^{m}, \qquad r=2,\ldots,n,\label{maciej}
\end{gather}
where $\alpha$ is a parameter, and where the potentials $V_{r}^{(k)}$, $k\in\mathbf{Z}$, are def\/ined by the following separation relations
\begin{gather}
\lambda_{i}^{k}+\sum_{j=1}^{n}\lambda_{i}^{n-j}V_{j}^{(k)}=0,\qquad i=1,\ldots,n.\label{pot}
\end{gather}
As such, the potentials $V_{r}^{(k)}$ are called basic separable potentials. Thus, the functions $h_{r}^{m}$ in~(\ref{maciej}) are generated by the
following quasi-separation relations
\begin{gather}
\alpha\lambda_{i}^{k}+\sum_{j=1}^{n}\lambda_{i}^{n-j}h_{j}=\frac{1}{2}\lambda_{i}^{m}\mu_{i}^{2}+\sum_{k=1}^{n}u_{ik}(\lambda)\mu_{k},\qquad i=1,\ldots,n,\label{qs}
\end{gather}
with $u_{ik}$ as in (\ref{basia}), i.e., the term $\lambda_{i}^{k}$ in~(\ref{qs}) generates the potential $V_{r}^{(k)}$ in~(\ref{maciej}).

\begin{Lemma}\label{111}The functions $H_{1}^{m},\ldots,H_{n}^{m},W_{2}^{m},\ldots,W_{n}^{m}$ constitute a maximally superintegrable $($with respect to $H_{1}^{m})$ system only in the following four cases: $(m,k)=(0,n),(n,n),(1,-1)$ and $(n+1,-1)$.
\end{Lemma}

\begin{proof}Since the functions $H_{r}^{m}$ are St\"{a}ckel Hamiltonians (albeit no longer geodesic), they commute: $\{ H_{r}^{m},H_{s}^{m}\} =0$ for any $r$, $s$, so $H_{1}^{m},\ldots,H_{n}^{m}$ constitute a Liouville integrable system. Moreover, $\{ E_{1}^{m},W_{r}^{m}\} =0$ for any $r$ due to Lemma~\ref{W}. Thus, in order to have $\{ H_{1}^{m},W_{r}^{m}\} =0$ for any $r$ we have to demand $\{ V_{1}^{(k)},W_{r}^{m}\} =0$ or, equivalently, $L_{Y_{r}^{m}}(V_{1}^{(k)})=0$, for all~$r$ which is satisf\/ied only in the mentioned four cases.
\end{proof}

Let us now turn to algebraic properties of the set of functions $h_{r}^{m}$ in~(\ref{maciej}). Since some of these functions can now commute to the constant~$\alpha$ we have to extend the set of function $h_{r}^{m}$ by the constant Hamiltonian $h_{0}=\alpha$ in order to turn it into a Poisson algebra.

\begin{Proposition}\label{prop}The functions $h_{0}^{m}=\alpha,h_{1}^{m},\ldots,h_{n}^{m}$ in~\eqref{maciej} for $(m,k)=(0,n),(n,n),(1,-1)$ and $(n+1,-1)$ constitute a~Poisson algebra with the commutation relations as in~\eqref{str}.
\end{Proposition}

Let us note that since the potentials $V_{r}^{(k)}$ in (\ref{maciej}) do not appear in the f\/irst part of Hamiltonian equations~(\ref{2}) they do not appear in the corresponding hydrodynamic system~(\ref{H}). Nor does the constant Hamiltonian $h_{0}^{m}=\alpha$. Therefore, the hydrodynamic system corresponding to~(\ref{h}) and the one corresponding to~(\ref{maciej}) are exactly the same.

Due to Lemma~\ref{111}, we will on what follows only need the potentials $V_{r}^{(n)}$ and $V_{r}^{(-1)}$. They can be calculated from~(\ref{pot}) and are in $q$-coordinates given by
\begin{gather*}
V_{r}^{(n)}=q_{r}, \qquad V_{r}^{(-1)}=\frac{q_{r-1}}{q_{n}}, \qquad r=1,\ldots,n.
\end{gather*}

Let us now perform the St\"{a}ckel transform of the Hamiltonians $h_{r}^{m}$ in~(\ref{maciej}) with respect to the parameter~$\alpha$. It means that we f\/irst solve the relation $h_{1}^{m}=\tilde{\alpha}$, i.e., $E_{1}^{m}+\alpha V_{1}^{(k)}=\tilde{\alpha}$, with respect to $\alpha$ which yields
\begin{gather*}
\tilde{h}_{1}^{m}=\alpha=-\frac{1}{V_{1}^{(k)}}E_{1}^{m}+\tilde{\alpha}\frac{1}{V_{1}^{(k)}}
\end{gather*}
and then replace $\alpha$ with $\tilde{h}_{1}^{m}$ in all the remaining Hamiltonians $h_{r}^{m}$, $r=2,\ldots,n$, which yields
\begin{gather*}
\tilde{h}_{r}^{m}= h_{r}^{m}\big\vert _{\alpha\rightarrow\tilde {h}_{1}^{m}}=E_{r}^{m}+\left( -\frac{1}{V_{1}^{(k)}}E_{1}^{m}+\tilde{\alpha
}\frac{1}{V_{1}^{(k)}}\right) V_{r}^{(k)}+W_{r}^{m}.
\end{gather*}
Thus, the St\"{a}ckel transform of (\ref{maciej}) with respect to $\alpha$ attains the form
\begin{gather}
\tilde{h}_{1}^{m} =-\frac{1}{V_{1}^{(k)}}E_{1}^{m}+\tilde{\alpha}\frac {1}{V_{1}^{(k)}},\label{h1}\\
\tilde{h}_{r}^{m} =E_{r}^{m}-\frac{V_{r}^{(k)}}{V_{1}^{(k)}}E_{1}^{m}+\tilde{\alpha}\frac{V_{r}^{(k)}}{V_{1}^{(k)}}+W_{r}^{m}=\tilde{H}_{r}
^{m}+W_{r}^{m},\qquad r=2,\dots ,n,\label{St}
\end{gather}
where
\begin{gather*}
\tilde{H}_{r}^{m}=\tilde{E}_{r}^{m}+\tilde{\alpha}\frac{V_{r}^{(k)}}{V_{1}^{(k)}},\qquad \tilde{E}_{r}^{m}=\frac{1}{2}p^{T}\tilde{K}_{r}\tilde{G}_{m}p.
\end{gather*}
Due to (\ref{h1}) the metric $\tilde{G}_{m}$ in $\tilde{h}_{1}^{m}$ is of the form
\begin{gather*}
\tilde{G}_{m}=-\frac{1}{V_{1}^{(k)}}G_{m}. 
\end{gather*}
Since $\tilde{G}_{m}$ is a conformal deformation of the metric $G_{m}$ by $V_{1}^{(k)}$, which in the considered four cases $(m,k)=(0,n),(n,n),(1,-1),(n+1,-1)$ satisf\/ies $L_{Y_{r}^{m}}(V_{1}^{(k)})=0$, all $n-1$ vector f\/ields~$Y_{r}^{m}$ are in these four cases Killing vectors for the metric~$\tilde{G}_{m}$ as well. Thus, the corresponding functions $W_{r}^{m}=p^{T}Y_{r}^{m}$ are constants of motion not only for~$H_{1}^{m}$ but also for $\tilde{H}_{1}^{m}$ and therefore the functions $\tilde{H}_{1}^{m},\ldots,\tilde{H}_{n}^{m},W_{2}^{m},\ldots,W_{n}^{m}$ also constitute a~maximally superintegrable system.

Due to the def\/inition of our St\"{a}ckel transform (\ref{St}), on the level of separation relations~(\ref{qs}) the St\"{a}ckel transform renders the following substitution:
\begin{gather}
\alpha\rightarrow\tilde{h}_{1}^{m},\qquad h_{1}\rightarrow\tilde{\alpha },\qquad h_{r}^{m}\rightarrow\tilde{h}_{r}^{m}\qquad \text{for} \quad r=2,\ldots,n.\label{sub}
\end{gather}
In consequence, the quasi-separation relations for our four cases have the form
\begin{gather}
\tilde{\alpha}\lambda_{i}^{n-1}+\lambda_{i}^{k}\tilde{h}_{1}^{m}+ \sum_{j=2}^{n}\lambda_{i}^{n-j}\tilde{h}_{j}^{m}=\frac{1}{2}\lambda_{i}^{m}\mu_{i}^{2}+\sum_{s=1}^{n}u_{is}(\lambda)\mu_{s},\qquad i=1,\ldots,n,\label{qRt}
\end{gather}
with $u_{ik}$ as in (\ref{basia}) (compare with~(\ref{qs})). However, applying the canonical transformation $\lambda_{i}\rightarrow\lambda_{i}^{-1}$, $\mu_{i}\rightarrow-\lambda_{i}^{2}\mu_{i}$ to the quasi-separation relations~(\ref{qRt}) with $(m,k)=(1,-1)$ one obtains the quasi-separation relations~(\ref{ola}) with $m=n+1$ while in the case $(m,k)=(n+1,-1)$ we receive after this transformation the quasi-separation relations~(\ref{ola}) with $m=1$ so they can not be considered as new systems. Let us thus focus on two remaining cases: $(m,k)=(0,n)$ and $(n,n)$. For both these cases the Killing tensors~$\tilde{K}_{r}$ can be expressed through the Killing tensors~$K_{r}$~(\ref{Kr}) as~\cite{maciej6}
\begin{gather}
\tilde{K}_{r}=K_{r+1}-K_{2}K_{r}, \qquad r=1,\ldots,n.\label{Kt}
\end{gather}

Let us now f\/ind whether the functions $\tilde{h}_{r}^{m}$ also constitute a~Poisson algebra. As it has been demonstrated in~\cite{maciej2}
\begin{gather}
\{\tilde{h}_{r},\tilde{h}_{s}\}=\sum_{i,j=1}^{n}\big( A^{-1}\big)_{ri}\big( A^{-1}\big) _{sj} \{ h_{i},h_{j} \}, \label{blee}
\end{gather}
where the $n\times n$ matrix $A$ is of the form
\begin{gather*}
A_{ij} =\delta_{ij},\qquad j=2,\ldots,n, \qquad A_{i1} =-\frac{\partial h_{i}}{\partial\alpha}=-V_{i}^{(k)},\qquad i=1,\ldots,n.
\end{gather*}
So, in general, there is no guarantee that the functions $\tilde{h}_{r}$ constitute an algebra even if~$h_{r}$ do.

\begin{Proposition}In the two cases $(m,k)=(0,n)$ and $(n,n)$ the functions $\tilde{h}_{r}^{m}$ constitute a~Poisson algebra, with the structure constants obtained from the structure constants of the algebra of~$h_{r}^{m}$ given in Proposition~{\rm \ref{prop}} through the substitution~\eqref{sub}.
\end{Proposition}

To see this, it is enough to realize that in both cases $(k,m)=(0,n)$ and~$(n,n)$ the formula~(\ref{blee}) reduces to
\begin{gather}
\big\{\tilde{h}_{r}^{m},\tilde{h}_{s}^{m}\big\}=\big\{ h_{r}^{m},h_{s}^{m}\big\} \qquad \text{for all} \ \ r,\,s,\label{zuzia}
\end{gather}
while the substitution (\ref{sub}) in the right-hand side of~(\ref{zuzia}) allows for expressing $\{h_{r}^{m},h_{s}^{m}\}$ in terms of the Hamiltonians~$\tilde{h}_{r}^{m}$.

\begin{Example}For $n=5$, $m=0$ and $k=n$ the Hamiltonians $h_{r}^{m}$ in~(\ref{maciej}) constitute a Poisson algebra with the following non-zero brackets:
\begin{gather*}
 \{ h_{2},h_{5} \} =3\alpha , \qquad \{ h_{3},h_{4} \} =\alpha , \qquad \{ h_{3},h_{5} \} =2h_{1}, \qquad \{h_{4},h_{5} \} =h_{2},
\end{gather*}
and after the St\"{a}ckel transform (\ref{St}) the corresponding algebra of~$\tilde{h}_{1}^{m}$ has due to~(\ref{sub}) the following non-zero brackets
\begin{gather*}
\big\{\tilde{h}_{2},\tilde{h}_{5}\big\}=3\tilde{h}_{1}, \qquad \big\{\tilde{h}_{3},\tilde{h}_{4}\big\}=\tilde{h}_{1}, \qquad \big\{\tilde{h}_{3},\tilde{h}_{5}\big\}=2\tilde{\alpha}, \qquad \big\{\tilde{h}_{4},\tilde{h}_{5}\big\}=\tilde{h}_{2}.
\end{gather*}
\end{Example}

Our two cases, that is $(m,k)=(0,n)$ or $(n,n)$, lead after the St\"{a}kel transform to two new non-homogeneous hydrodynamic systems through the Theorem~\ref{jedyne}. They have the form
\begin{gather}
q_{t_{r}}=\tilde{K}_{r}q_{x}+Y_{r}\equiv\tilde{Z}_{r}, \qquad r=2,\ldots,n,\label{sylwia}
\end{gather}
with $Y_{r}$ for $m=0$ and $m=n$ the same as before and with $\tilde{K}_{r}$ given by~(\ref{Kt}).

\begin{Example}For $n=4$ the systems (\ref{sylwia}) attain in Vi\`{e}te coordinates $q$ the form
\begin{gather*}
\left[
\begin{matrix}
q_{1}\\
q_{2}\\
q_{3}\\
q_{4}
\end{matrix}
\right] _{t_{2}}=\left[
\begin{matrix}
q_{2} & -q_{1} & 0 & 0\\
q_{1}q_{2} & q_{2}-q_{1}^{2} & -q_{1} & 0\\
q_{1}q_{3} & 0 & q_{2}-q_{1}^{2} & -q_{1}\\
q_{1}q_{4} & 0 & 0 & q_{2}-q_{1}^{2}
\end{matrix}
\right] \left[
\begin{matrix}
q_{1}\\
q_{2}\\
q_{3}\\
q_{4}
\end{matrix}
\right] _{x}+Y_{2}=\tilde{Z}_{2},
\\
\left[
\begin{matrix}
q_{1}\\
q_{2}\\
q_{3}\\
q_{4}
\end{matrix}
\right] _{t_{3}}=\left[
\begin{matrix}
q_{3} & 0 & -q_{1} & 0\\
q_{1}q_{3} & q_{3} & -q_{1}^{2} & -q_{1}\\
q_{1}q_{4} & q_{1}q_{3} & q_{3}-q_{1}q_{2} & -q_{1}^{2}\\
0 & q_{1}q_{4} & 0 & q_{3}-q_{1}q_{2}
\end{matrix}
\right] \left[
\begin{matrix}
q_{1}\\
q_{2}\\
q_{3}\\
q_{4}
\end{matrix}
\right] _{x}+Y_{3}=\tilde{Z}_{3},
\\
\left[
\begin{matrix}
q_{1}\\
q_{2}\\
q_{3}\\
q_{4}
\end{matrix}
\right] _{t_{4}}=\left[
\begin{matrix}
q_{4} & 0 & 0 & -q_{1}\\
q_{1}q_{4} & q_{4} & 0 & -q_{1}^{2}\\
0 & q_{1}q_{4} & q_{4} & -q_{1}q_{2}\\
0 & 0 & q_{1}q_{4} & q_{4}-q_{1}q_{3}
\end{matrix}
\right] \left[
\begin{matrix}
q_{1}\\
q_{2}\\
q_{3}\\
q_{4}
\end{matrix}
\right] _{x}+Y_{4}=\tilde{Z}_{4},
\end{gather*}
where
\begin{alignat*}{3}
& \text{for }m =0\colon \quad && Y_{2}=(0,0,0,1)^{T},\qquad Y_{3}=(0,0,2,q_{1})^{T},\qquad Y_{4}=(0,3,2q_{1},q_{2})^{T}, &\\
& \text{for }m =4\colon \quad && Y_{2}=(0,q_{2},2q_{3},3q_{4})^{T},\qquad Y_{3}=(0,q_{3},2q_{4},0)^{T},\qquad Y_{4}=(0,q_{4},0,0)^{T},&
\end{alignat*}
which constitute Lie algebras with the following nonzero elements given by Theorem~\ref{jedyne} \linebreak and~(\ref{str})
\begin{alignat*}{3}
& m=0\colon \quad && [\tilde{Z}_{2},\tilde{Z}_{4}]=-2\tilde{Z}_{1},& \\
& m=4\colon \quad && [\tilde{Z}_{2},\tilde{Z}_{3}]=\tilde{Z}_{3},\qquad [\tilde{Z}_{2},\tilde{Z}_{4}]=2\tilde{Z}_{4}, &
\end{alignat*}
where $\tilde{Z}_{1}=(q_{1},q_{2},q_{3},q_{4})_{x}^{T}$.
\end{Example}

\pdfbookmark[1]{References}{ref}
\LastPageEnding

\end{document}